\documentclass[11pt]{article}
\usepackage{style}
\graphicspath{ {./images/} }

\usepackage{blindtext}
\usepackage{algorithm,algpseudocode}
\usepackage{graphicx}
\usepackage[round]{natbib}
\usepackage{subfiles}
\usepackage{authblk}
\usepackage{cleveref}
\usepackage{dsfont}
\usepackage{mathtools}
\usepackage{multirow}

\title{De-biased Two-Sample U-Statistics With Application To Conditional Distribution Testing}

\begin{document}
\author[1]{Yuchen Chen}
\author[1]{Jing Lei}
\affil[1]{Carnegie Mellon University}

\maketitle

\begin{abstract}
In some high-dimensional and semiparametric inference problems involving two populations, the parameter of interest can be characterized by two-sample U-statistics involving some nuisance parameters.  In this work we first
extend the framework of one-step estimation with cross-fitting to two-sample U-statistics, showing that using an orthogonalized influence function can effectively remove the first order bias, resulting in asymptotically normal estimates of the parameter of interest.  As an example, we apply this method and theory to the problem of testing two-sample conditional distributions, also known as strong ignorability.  When combined with a conformal-based rank-sum test, we discover that the nuisance parameters can be divided into two categories, where in one category the nuisance estimation accuracy does not affect the testing validity, whereas in the other the nuisance estimation accuracy must satisfy the usual requirement for the test to be valid.  We believe these findings provide further insights into and enhance the conformal inference toolbox.
\end{abstract}

\section{Introduction}



Many statistical inference problems involve the estimation of functionals in the presence of nuisance parameters. In the basic setup, we observe independent data $Z_1,...,Z_n \sim P$ from some underlying distribution $P$ in some model class $\mathcal{P}$. We would like to estimate some functional of the form $\theta = E_{P} \varphi(Z,\gamma)$, where $\varphi$ is a known function but $\gamma$ is an unknown nuisance parameter. A natural way to estimate such a functional is through the following two-step procedure. First, we use part of the data to estimate the nuisance parameter $\hat{\gamma}_n$. Given such an estimate, we can take expectation over the empirical distribution of our data to get an estimate
\begin{equation*}
    \hat{\theta}_n = \frac{1}{n} \sum_{i=1}^n \varphi(Z_i,\hat{\gamma}_n).
\end{equation*}

Such estimates are called plug-in estimates. When the nuisance parameter $\gamma$ is highly complex, such as when it is high/infinite-dimensional, the first stage will often involve using flexible or regularized machine learning methods such as LASSO, random forests, or neural nets. It is known that when such methods are used for estimating the nuisance parameter, the plug-in estimate will be biased and not achieve desired asymptotic properties leading to invalid inference \citep{chernozhukov_doubledebiased_2018}.

Double/debiased machine learning offers a general procedure for estimating functionals in the presence of complex nuisance parameters \citep{chernozhukov_doubledebiased_2018}. The idea is to orthogonalize our original functional by adding a carefully chosen mean-zero function $\phi(Z,\gamma,\alpha)$, where $\alpha$ is a possible additional nuisance parameter. That is, we now estimate
\begin{equation*}
    \theta = E_P [\varphi(Z,\gamma) + \phi(Z,\gamma,\alpha)].
\end{equation*}
With this orthogonalization, the bias is now second-order in terms of the estimation error of $\gamma$ and $\alpha$. In particular, when $\gamma$ and $\alpha$ are estimated at faster than $n^{-1/4}$ rates, we can still expect the estimate of $\theta$ to be root-$n$ consistent and asymptotically normal. These rates are achievable by flexible and regularized estimators and thus we can still achieve correct inference with highly complex nuisance parameters. 

Recently, \citet{escanciano_debiased_2023} extended the double/debiased machine learning method and theory to degree-$2$ one-sample U-statistics. In this case, we are interested in estimating a functional of the form $E[\varphi(X_i,X_j,\gamma)]$, where $\varphi$ is symmetric in $X_i$ and $X_j$ and $\gamma$ is a nuisance parameter. Similarly, the plug-in estimate
\begin{equation*}
    \hat{\theta}_n = \binom{n}{2}^{-1} \sum_{i < j} \varphi(X_i,X_j,\hat{\gamma}_n),
\end{equation*} 
may be highly biased and lead to invalid inference when the nuisance parameter $\gamma$ is estimated using regularized or flexible ML methods. \citet{escanciano_debiased_2023} develop a U-statistic influence function that allows the double/debiased machine learning framework to be applied to two sample U-statistics as well. Accordingly, by correcting with the U-statistic influence function, we can achieve root-$n$ consistent and asymptotically normal estimators under mild regularity conditions.

One setting where the double/debiased machine learning framework is missing is in two-sample problems. In this case, we observe two samples $X_1,...,X_m \sim P$ and $U_1,...,U_n \sim Q$ and we would like to estimate 
\begin{equation*}
    \U \varphi(X,U,\gamma) := E_{P,Q} \varphi(X,U,\gamma).
\end{equation*} 
Given a first stage estimate $\hat{\gamma}_{mn}$, we can form the plug-in estimate
\begin{equation*}
    \U_{m,n} \varphi(X,U,\hat{\gamma}_{mn}) := \frac{1}{mn}\sum_{i=1}^m \sum_{j=1}^n \varphi(X_i,U_j,\hat{\gamma}_{mn}).
\end{equation*}
When $\gamma$ is estimated using flexible nonparametric methods, we again run into the same issues as in the one-sample case. In this paper, we extend the construction of U-statistic influence functions for one-sample U-statistics to two-sample U-statistics. In particular, we show how to use these influence functions to correct for bias induced by flexible estimation of complex nuisance parameters resulting in asymptotic normal estimators under mild regularity conditions.

An important instance where two-sample U-statistics appear is in two-sample testing problems. We consider the following two-sample testing problem. Given two samples
\begin{equation*}
    \{(X_i,Y_i)\}_{i=1}^m \sim F \text{ and } \{(U_j,V_j)\}_{j=1}^m \sim G
\end{equation*}
on the same sample space,
we are interested in testing whether
\begin{equation}
    \label{hypothesis}
    H_0: F_{Y|X} = G_{V|U} \text{ versus } H_1: F_{Y|X} \neq G_{V|U},
\end{equation}
where $F_{Y|X}$ and $G_{V|U}$ are the conditional distributions of $Y$ and $V$ given $X$ and $U$ respectively. The null hypothesis \ref{hypothesis} is commonly referred to as the covariate shift assumption \citep{hu_two-sample_2023}, and is also closely related to the ``strong ignorability'' condition in causal inference \citep{rosenbaum_central_1983}.

Transfer learning is a domain where the covariate shift assumption plays a crucial role \citep{pan_survey_2010}. In the previous setting, let $X_i$ and $U_j$ denote covariates and $Y_i$ and $V_j$ denote the response. Suppose that the first sample, $\{(X_i,Y_i)\}_{i=1}^m$ is our training data. Using this training data, we can build a model $\hat{f}$ such that $\hat{f}(X_i)$ is a prediction of $Y_i$. Now suppose we have test data $\{U_j\}_{j=1}^n$ on the same sample space but drawn from a different distribution. Transfer learning investigates whether it is possible to use $\hat{f}$ built using observations $(X_i,Y_i)$ to predict $V_j$ from $U_j$. 

Covariate shift is one setting where transfer learning is possible. In the covariate shift setting, as long as we can understand the marginal density ratio between the two samples, transfer learning can be achieved. To safely apply these techniques, it is important to verify that the covariate shift assumption holds. Transfer learning is useful in practice as it may be difficult to obtain labeled test data for one sample, but such test data is plentiful in a similar sample. With a covariate shift test, we can use the available labeled data from both samples to check if the covariate shift assumption is satisfied. If so, we can safely use models trained on the sample with abundant labeled training data for prediction on the other sample using transfer learning techniques.

Using ideas from weighted conformal prediction, \citet{hu_two-sample_2023} introduce a hypothesis test for covariate shift. Suppose the second sample consists of a single point $(U,V)$. Then weighted conformal prediction says that
\begin{equation*}
    \frac{1}{m} \sum_{i=1}^m \gamma(X_i)I\left(s(X_i,Y_i) < s(U,V)\right)
\end{equation*}
is a valid p-value (disregarding ties), where $\gamma:=\frac{g_U}{f_X}$ is the $X,U$ marginal density ratio of $F$ and $G$ and $s$ is any non-degenerate conformity score function. To extend this p-value to incorporate multiple points, \citet{hu_two-sample_2023} aggregate these conformal p-values into a two-sample U-statistic

\begin{equation*}
    \frac{1}{mn} \sum_{i=1}^m \sum_{j-1}^n \gamma(X_i)I\left(s(X_i,Y_i) < s(U_j,V_j)\right).
\end{equation*}
which is used as the test statistic. Under the null and any non-degenerate choice of $s$, the expected value of this test statistic is $\frac{1}{2}$. Under the alternative and when $s$ is the conditional density ratio, the expected value is strictly less than $\frac{1}{2}$. Thus, we can use this test statistic to give a one-sided hypothesis test for covariate shift.

Remarkably, this covariate shift test makes minimal assumptions on the distributions $F,G$, a property inherited from conformal prediction, and has power guarantees against any alternative with a good choice of conformal score function $s$. However, the validity of this test statistic is dependent on the estimation of a complex (possibly high-dimensional or nonparametric) nuisance parameter $\gamma$. Thus, the test statistic fits into the two-sample debiased U-statistic framework. By correcting using the two-sample U-statistic influence function, we derive a novel test for covariate shift based on aggregating conformity scores.


\section{Debiased Two-Sample U-Statistics}



In this section, we define the condition of Neyman orthogonality and show how to augment the original estimating equation by influence functions to achieve Neyman orthogonality.

For the next two sections, we will use the following general notation that extends beyond the regression/classification setting. Suppose that we observe two samples $X_1,...,X_m \sim F_0$ and $U_1,...,U_n \sim G_0$. For a function $\varphi(X,U,\gamma)$, we are interested in estimating 
\begin{equation*}
\theta_0 = \U[\varphi(X,U,\gamma_0)]\,,
\end{equation*}
where $\U$ denotes the population U-statistic that takes expectation over the product distribution of $(X,U)$, and the term $\gamma_0=\gamma(F_0,G_0)$ is the true value of a nuisance parameter, the value of a bivariate functional $\gamma(\cdot,\cdot)$ evaluated at $(F_0,G_0)$.

When $\gamma_0$ is estimated nonparametrically using popular machine learning methods, the resulting estimate of $\theta_0$ is known to be biased unless the estimate of $\gamma_0$ achieves an unrealistically high level of accuracy. In this section, we define orthogonality which eliminates the first order bias introduced through estimating $\gamma$, and show how to augment the original moment function by influence functions to achieve orthogonality.

\subsection{Neyman Orthogonality}

Neyman Orthogonality refers to removing the first order bias of the effect of the nuisance $\gamma$.  To gain some intuition, assume for now $\gamma_0=\gamma(F_0,G_0)$ is a scalar, and $\hat\gamma_0=\gamma(G_n,F_n)$ is an estimate. Assume $\theta(\gamma)=\U[\varphi(X,U,\gamma)]$ is a smooth function from $\mathbb R$ to $\mathbb R$. The plug-in estimate
$\hat\theta_n$ targets the perturbed version $\theta(\hat\gamma_0)$, which has bias
$\theta(\hat\gamma_0)-\theta(\gamma_0)=\theta'(\gamma_0)(\hat\gamma_0-\gamma_0)+o((\hat\gamma_0-\gamma_0)^2)$.  If $\theta'(\gamma_0)\neq 0$, then the bias has the same order as $\hat\gamma_0-\gamma_0$, which is typically larger than $n^{-1/2}$ if $\hat\gamma_0$ is obtained nonparametrically. The idea of orthogonalization is to modify $\varphi$ to a different U-statistic kernel that has the same mean value but zero first order derivative with respect to the nuisance parameter $\gamma$.

To make this idea rigorous and concrete, we need to define functional derivatives in our scenario. To this end, we will resort to path-wise differentiation, i.e., differentiation along 1-dimensional paths in the functional space. This type of functional derivative is closely related to the von Mises expansion (functional Taylor expansion), first considered by  \citet{v_mises_asymptotic_1947} and is a tool widely used in semiparametric theory \citep{pfanzagl_contributions_1982}.

Given alternate distributions $H,K$, we define a path $(F_\epsilon,G_\epsilon)_{\epsilon\in[0,1]}$, where
\begin{itemize}
    \item $F_\epsilon = (1-\epsilon)F_0 + \epsilon H$
    \item $G_\epsilon = (1-\epsilon)G_0 + \epsilon K$.
\end{itemize}
This is a path that starts at the true distribution $(F_0,G_0)$ and ends at $(H,K)$.

\begin{definition}
The U-statistic kernel $\varphi$ is Neyman orthogonal if for any one-dimensional path $(F_\epsilon,G_\epsilon)$, we have
\begin{equation*}
    \frac{d}{d\epsilon} 
    \U[\varphi(X,U,\gamma_\epsilon)] \big|_{\epsilon = 0} = 0.
\end{equation*}
\end{definition}
Intuitively, Neyman orthogonal kernels are those whose first order derivative with respect to the nuisance parameter equals $0$ and hence the plug-in estimate has only higher order bias. Next we describe how to orthogonalize an arbitrary kernel.

\subsection{U-Statistic Influence Function}
Suppose that our U-statistic kernel $\varphi$ is not orthogonal.
 In this case, we would like to find a new kernel $\psi$ such that $\U \psi(X,U,\gamma_0) = \theta_0$ and is also orthogonal. In this section, we will see how to accomplish this by augmenting $\varphi$ with a U-statistic influence function. This is an extension of the influence functions used in \citet{chernozhukov_doubledebiased_2018} for one-sample problems and follows from the influence functions used in \citet{escanciano_debiased_2023} for one-sample U-statistics. 

The orthogonalization of $\varphi$ may involve some additional nuisance parameters, denoted by $\alpha_0=\alpha(F_0,G_0)$.
We say $\phi(X,U,\gamma,\alpha)$ is a U-Statistic influence function if for any path $(F_\epsilon,G_\epsilon)$ from $(F_0,G_0)$ to $(H,K)$ as defined earlier, $\phi$ satisfies the following properties
\begin{equation}
    \label{mean zero}
    \int\int \phi(x,u,\gamma_\epsilon,\alpha_\epsilon) dF_\epsilon(x) dG_\epsilon(u) = 0 \text{ for } \epsilon \in [0,\overline{\epsilon}]
\end{equation}
and
\begin{equation}
    \label{pathwise diff}
    \frac{d}{d\epsilon}\U[\varphi(X,U,\gamma_\epsilon)] \big|_{\epsilon = 0} = \int\int \phi(x,u,\gamma_0,\alpha_0) [dF_0(x)dK(u)+dH(x)dG_0(u)].
\end{equation}
Here $\alpha_\epsilon$ is defined in the same way as $\gamma_\epsilon$, and $\overline{\epsilon} >0$ is a constant small enough such that both $\gamma_\epsilon$ and $\alpha_\epsilon$ exists for $\epsilon \in [0,\overline{\epsilon}].$

If such an influence function exists, note that condition \ref{mean zero} with $\epsilon = 0$ shows that
\begin{equation*}
    \U_{F_0,G_0} \phi(X,U,\gamma_0,\alpha_0) = 0.
\end{equation*}
Then the debiased U-statistic kernel
\begin{equation}
    \label{debias moment}
    \psi(X,U,\gamma,\alpha) := \varphi(X,U,\gamma) + \phi(X,U,\gamma,\alpha) 
\end{equation}
satisfies $\U\psi(X,U,\gamma_0,\alpha_0) = \theta_0$. We claim this moment function is Neyman orthogonal, as given in the following theorem, which is proved in \Cref{app:proof_sec3}.

\begin{theorem}\label{thm:neyman_ortho}
The U-statistic kernel $\psi$ in \eqref{debias moment} with influence function $\phi$ satisfying \eqref{mean zero} and \eqref{pathwise diff} is Neyman orthogonal. That is, it satisfies
\begin{equation*}
    \frac{d}{d\epsilon}\U[\psi(X,U,\gamma_\epsilon,\alpha_\epsilon)]\big|_{\epsilon=0} = 0\,.
\end{equation*}
\end{theorem}

\subsection{Cross-Fitting}
In order to implement the debiased estimator, we need to first estimate the nuisance parameters $\gamma$ and $\alpha$. Estimating the nuisance parameters and evaluating the empirical estimate on the same data sample is a form of double-dipping which can lead to incorrect inference if not carefully handled. There have been two popular approaches to address this issue. The first is to use empirical process techniques, such as by imposing Donsker conditions on the estimation of the nuisance parameters. However, Donsker conditions are difficult to verify when we use black-box ML methods. Thus, we will not discuss these conditions any further. For additional references on this approach see \cite{bolthausen_lectures_2002} (Section 6).

Naturally, another way to remove the double-dipping is to sample split. That is use half the sample to estimate the nuisance parameters and the other half to estimate the parameter. Sample splitting works for any method of estimation for the nuisance parameters but has the downside of only making use of part of the data for estimation. To increase the efficiency of sample splitting, cross-fitting was introduced. The basic idea of cross-fitting is that we will swap the roles of the data used for nuisance and parameter estimation and then aggregate. The use of cross-fitting to avoid restrictive empirical process techniques has its roots in functional estimation and is widely used in modern semiparametric methods such as double/debiased machine learning \citep{bickel_estimating_1988,hasminskii_asymptotically_1986,pfanzagl_contributions_1982,schick_asymptotically_1986,chernozhukov_doubledebiased_2018,chernozhukov_locally_2022,escanciano_debiased_2023,kennedy_semiparametric_2023}. We need to adapt cross-fitting for two-sample settings. 


For any positive integer $m$, define the set $[m]:=\{1,...,m\}$. We first partition the $x,u$ pairs. Let $C_1,...,C_S$ be a partition of $[m]$ and $D_1,...,D_T$ be a partition of $[n]$. We will use $I_{st} := C_s \times D_t$ to index the pairs $(X_i,U_j)$ where $i \in C_s$ and $j \in D_t$. We let $M_s:= |C_s|$, the cardinality of $C_s$ and $N_t := |D_t|$.

For each $s \in [S]$ and $t \in [T]$, we have an estimate
\begin{equation}\label{eq:theta_st}
    \hat{\theta}_{st} = \frac{1}{M_s N_t} \sum_{i,j \in I_{st}} \psi(X_i,U_j,\hat{\gamma}_{-st},\hat{\alpha}_{-st}).
\end{equation}
We use $\hat{\gamma}_{-st}$ and $\hat{\alpha}_{-st}$ to denote estimates of $\gamma$ and $\alpha$ using observations $\{(X_i,U_j):i \not\in C_s \text{ and } j \not\in D_t\}$.

The estimates on the individual folds can be aggregated into a single estimate by
\begin{equation}\label{eq:theta_crossfitting}
    \hat{\theta} = \sum_{s \in [S]} \sum_{t \in [T]} \frac{M_s N_t}{mn} \hat{\theta}_{st}.
\end{equation}

Note that when all folds are of equal size, this is just the mean of the estimates across all folds. 

\subsection{Finding the nuisance influence function $\phi$}
The nuisance influence function $\phi$ can usually be found by expanding the differentiation in the left hand side of identity \eqref{pathwise diff} and matching the two sides. 

The left side of \eqref{pathwise diff} involves the choice of arbitrary distributions $H$ and $K$. Choosing $H$ and $K$ to be point-mass contaminants, i.e. $H=\delta(x)$ where $\delta(x)$ is the distribution that is $x$ with probability 1, can make computing the differentiation much easier. This is similar to the idea of Gateaux derivatives. Such computations will give a candidate influence function. To show it is an actual influence function, we must then check that \eqref{mean zero} and \eqref{pathwise diff} hold for general choices of $H$ and $K$. We will use this strategy to compute the influence function in the application to covariate shifts. Further examples of this strategy in the one-sample case can be found in \citet{kennedy_semiparametric_2023}. 

In general, there are other strategies for computing influence functions in the one sample case that could be applicable in the two sample case as well \citep{kennedy_semiparametric_2023,chernozhukov_doubledebiased_2018,chernozhukov_locally_2022,escanciano_debiased_2023}.


\section{Asymptotics of the de-biased estimate}



In this section, we show asymptotic normality of the debiased estimator $\hat\theta$ using the orthogonalized kernel with cross-fitting as given in \eqref{eq:theta_crossfitting}. We keep the same notation as in the previous section. To study the asymptotic behavior of our estimate, we will consider a sequence of models indexed by $(m,n)$, where the sample sizes $(m,n)$ grow to infinity. 

For a given $s,t$, we have the expansion
\begin{equation*}
    \begin{split}
        \hat{\theta}_{st} - \theta_0 &= \hat{\mathbb{U}}_{st} \hat{\psi} - \mathbb{U}\psi\\
        &= \mathbb{U}_{st}(\hat{\psi}_{st} - \psi) + (\mathbb{U}_{st} - \mathbb{U})\psi\\
        &= (\hat{\mathbb{U}}_{st}-\mathbb{U})\psi + (\hat{\mathbb{U}}_{st}-\mathbb{U})(\hat{\psi}_{st} - \psi) + \mathbb{U}(\hat{\psi}_{st} - \psi)\\
        &=: A_{st} + B_{st} + C_{st},
    \end{split}
\end{equation*}

where $\hat{\mathbb{U}}_{st}$ is the $U$-statistic operator over the empirical distribution on the observations indexed by $I_{st}$ and $\hat{\psi}_{st} = \psi(\cdot,\cdot,\hat{\gamma}_{-st},\hat{\alpha}_{-st}).$

Summing over all folds, it follows that
\begin{equation}
\label{decomposition}
    \begin{split}
        \hat{\theta} - \theta_0 = \sum_{s \in [S]} \sum_{t \in [T]}\frac{M_s N_t}{mn} A_{st} + \sum_{s \in [S]} \sum_{t \in [T]}\frac{M_s N_t}{mn} B_{st} + \sum_{s \in [S]} \sum_{t \in [T]}\frac{M_s N_t}{mn} C_{st}.
    \end{split}
\end{equation}

We analyze these sums separately. 

\subsection*{Term 1: $\sum_{s \in [S]} \sum_{t \in [T]}\frac{M_s N_t}{mn} A_{st}$}
The first term evaluates to 
\begin{equation*}
    \sum_{s \in [S]} \sum_{t \in [T]}\frac{M_s N_t}{mn} A_{st} = (\hat{\mathbb{U}} - \mathbb{U}) \psi,
\end{equation*}
where $\hat{\mathbb{U}}$ is the U-statistic operator over the empirical distribution of the entire sample. Note that $\sqrt{m+n} (\hat{\mathbb{U}} - \mathbb{U}) \psi$ will satisfy a U-statistic central limit theorem, provided the kernel $\psi$ is non-degenerate, and the sample sizes are not too imbalanced.  In particular, we make the following assumption on the sample sizes.
\begin{assumption}
\label{as: sample ratio}
    As $m,n \rightarrow \infty$, $ \lim_{m,n} \frac{m}{m+n} =\lambda\in(0,1)$.
\end{assumption}

Under the three-term decomposition \eqref{decomposition} and the asymptotic normality of the first term, we would like terms 2 and 3 be asymptotically negligible. Then $\sqrt{m+n}(\hat{\theta}-\theta_0)$ will be asymptotically equivalent to a term 1 which is asymptotically normal. 

\subsection*{Term 2: $\sum_{s \in [S]} \sum_{t \in [T]}\frac{M_s N_t}{mn} B_{st}$}
The second term is the empirical process term. If $\hat{\psi}_{st}$ was a fixed function, each $B_{st}$ would be a sample average of a mean-zero quantity with a variance ${\rm Var}(\hat\psi_{st}-\psi)M_sN_t/(mn)$ which is itself vanishing.  A technical challenge here is that the function $\hat\psi_{st}$ is random.  Without sample splitting, we would need empirical process theory to deal with the dependence between $\hat U_{st}$ and $\hat\psi_{st}$, such as assuming the random functions $\hat{\psi}_{st}$ are in a Donsker class \citep[Theorem 6.15]{bolthausen_lectures_2002}. 
However, Donsker conditions can be quite restrictive and hard to verify for many popular black-box style machine learning methods.

The other popular approach to handling this term, which we will adapt, uses sample splitting (specifically cross-fitting) \citep{chernozhukov_doubledebiased_2018,kennedy_semiparametric_2023,newey_cross-fitting_2018}. In our construction above, the randomness $\hat{\mathbb U}_{st}$ and $\hat\psi_{st}$ are from disjoint subsets of the data, hence allowing for a simple conditional argument. Through the use of cross-fitting, under mild consistency conditions for estimating the nuisance parameters, this term will be asymptotically negligible.  With the cross-fitting method, we only need the estimated orthognalized influence functions to be consistent in $L^2$ norm.
\begin{assumption}
\label{as: consistency}
    For every fold, indexed by $(s,t)$, we have
    \begin{equation*}
        \|\hat{\psi}_{st} - \psi\| = o_p(1)\,,
    \end{equation*}
    where $\|\cdot\|$ denotes the $L^2$-norm under the product distribution distribution $(X,U)\sim F\times G$.
\end{assumption}

\subsection*{Term 3: $\sum_{s \in [S]} \sum_{t \in [T]}\frac{M_s N_t}{mn} C_{st}$}

Term 3 is exactly the bias of our estimator. If the $U$-statistic is debiased, then this bias term will be second-order. In this case, this term will be asymptotically negligible after scaling as long as the product of the convergence rates of $\|\hat{\gamma} - \gamma\|$ and $\|\hat{\alpha} - \alpha\|$ is $\sqrt{m+n}$. For example, this occurs when both are faster than $o((m+n)^{-\frac{1}{4}})$. Such rates can be achieved by many flexible ML methods.  Although the orthogonality established in \Cref{thm:neyman_ortho} suggests that in general the bias of an orthogonalized $U$-statistic kernel will be of higher order.  To rigorously verify that it is asymptotically negligible so that the first term dominates is non-trivial, and needs to be done in a case-by-case manner.  Here we state a general condition on the third term required for asymptotic normality of $\hat\theta$. We will come back to verify it in our example of covariate shift testing in the next section.
\begin{assumption}
\label{as: bias}
    The term $\sum_{s \in [S]} \sum_{t \in [T]}\frac{M_s N_t}{mn} C_{st}$ is $o_p\left(\frac{1}{\sqrt{m+n}}\right)$.
\end{assumption}

\begin{lemma}
\label{lm: asymptotic}
   Under assumptions \ref{as: sample ratio}, \ref{as: consistency}, and \ref{as: bias}, it follows that 
\begin{equation*}
    \sqrt{m+n}(\hat{\theta} - \theta_0) = \sqrt{m+n}(\mathbb{U}_{mn} - \mathbb{U})\psi + o_p(1).
\end{equation*}
\end{lemma}

\Cref{lm: asymptotic} provides the main theoretical support for our method.  
A proof is given in \Cref{app:proof_sec3}.
A simple application of U-statistic CLT and Lemma \ref{lm: asymptotic} gives the desired asymptotic normality.

\begin{theorem}
    \label{thm: asymptotic normality}
    Given a two-sample U-statistic kernel $\psi(X,U,\gamma,\alpha)$ satisfying assumptions  \ref{as: sample ratio}, \ref{as: consistency}, and \ref{as: bias}, we have
    \begin{equation*}
        \sqrt{m+n}(\hat{\theta}-\theta_0) \xrightarrow{D} N(0,V),
    \end{equation*} 
    with
    \begin{equation*}
        V = \frac{cov(\psi(X,U),\psi(X,U'))}{\lambda} + \frac{cov(\psi(X,U),\psi(X',U))}{1-\lambda},
    \end{equation*}
    where $X'$ and $U'$ are iid copies of $X$ and $U$.
\end{theorem}
\begin{proof}
    By lemma \ref{lm: asymptotic}, we know that $ \sqrt{m+n}(\hat{\theta}-\theta_0)$ is asymptotically equivalent to $\sqrt{m+n} (\U_{mn}-\U)\psi$. 

    By U-statistic central limit theorem (see theorem 12.6 in \cite{vaart_asymptotic_1998}), if $E\psi^2(X,U) < \infty$, then the limiting distribution of $\sqrt{m+n} (\U_{mn}-\U)\psi$ is $N(0,V)$ with $V$ as defined in the theorem statement.
\end{proof}

Constructing confidence intervals or hypothesis tests using the estimates requires a consistent estimator of the variance of the U-statistic.
By the conditional decomposition of total variance, and the fact that ${\rm Cov}[\psi(X,U),\psi(X,U')|X]=0$, we have
$$
{\rm Cov}(\psi(X,U),\psi(X,U'))={\rm Var}(E(\psi(X,U)|X))\,.
$$
This allows us to estimate the variance by working with the variance of the standard projection of the U-statistic $\sqrt{m+n}(\mathbb{U}_{mn} - \mathbb{U})\psi$ by

\begin{equation*}
    \frac{m+n}{n}\var(E[\psi(X,U,\gamma,\alpha)|X]) + \frac{m+n}{m}\var(E[\psi(X,U,\gamma,\alpha)|U]) \rightarrow V.
\end{equation*}

Let $\psi_1(x)=E[\psi(X,U)|X=x]$.  
We can use plug-in estimates to estimate ${\rm Var}(\psi_1(X))$ using the previous cross-fit splits. For example, we can use $$\hat\psi_1(x_i)=\frac{1}{n} \sum_{t=1}^T \sum_{j \in D_t} \psi(x_i,u_j,\gamma_{-st},\alpha_{-st})\,.$$ 
Then we can  take $\hat\sigma_1^2$ as  the empirical variance  of  $\{\hat\psi_1(x_i):1\le i\le m\}$. A similar idea can be applied to obtain $\hat\sigma_2^2$ as the empirical variance of $\hat\psi_2(u_i)$, where $\psi_2(u)=E(\psi(X,U)|U=u)$.  The final asymptotic variance estimate is
$$
\hat V = \frac{m+n}{n}\hat\sigma_1^2+\frac{m+n}{m} \hat\sigma_2^2\,.$$
Under some standard consistency conditions on $\alpha_{-st}$ and $\gamma_{-st}$, it is possible to obtain consistency of $\hat V$ as an estimate of $V$.  The detailed argument will be very similar to those in \citet{escanciano_debiased_2023} for the one-sample case, and is omitted here.

\section{Two-sample conditional distribution test}



With the general theory now developed, we turn to an application to two-sample conditional distribution testing. Suppose we observe two samples, $(X,Y)\sim F$ and $(U,V) \sim G$. We use $X,U$ to denote the covariates of the two separate samples and $Y,V$ the corresponding response. So in this section, the paired random vector $(X,Y)$ corresponds to $X$ in the previous sections, and $(U,V)$ corresponds to $U$ in the previous sections. Let $f,g$ denote the marginal density of $X$ and $U$ respectively. We will use $f_{Y|X}$ and $g_{V|U}$ to denote the corresponding conditional distributions.

We are interested in whether the distribution of the response conditional on the covariates is equal. That is we are interested in testing
\begin{equation*}
    H_0 :~ f_{Y|X} = g_{V|U}\,,   ~~{\text vs }~~ H_1:~ f_{Y|X}\neq g_{Y|U}\,.
\end{equation*}
This hypothesis is of general interest in several important areas of statistical learning, including transfer learning \citep{sugiyama_machine_2012}, predictive inference \citep{peters_causal_2016}, and is closely related to the notion of strong ignorability in causal inference \citep{rosenbaum_central_1983}.

Using ideas from weighted conformal prediction, \citet{hu_two-sample_2023} proposed the test statistic
\begin{equation*}
    \U\varphi(X,Y,U,V,\gamma) = \U\left[\gamma(X)a(X,Y,U,V)\right].
\end{equation*}
where $$\gamma(X)=\frac{g(X)}{f(X)}$$ 
is the marginal density ratio of the covariates under  $G$ and $F$, and 
$$a=\mathds{1}(s(X,Y)<s(U,V))$$ where $s$ is some known function. 
Let 
$$\theta=\U\varphi(X,Y,U,V,\gamma)\,.$$
It is shown in \citet{hu_two-sample_2023}
that
\begin{align*}
    \theta = 1/2 &\text{ under } H_0\,,\\
    \theta < 1/2 & \text{ under } H_1\,.
\end{align*}
if a random tie-breaking is used in the function $a(\cdot)$ to avoid degeneracy.  In fact, the quantity $1/2-\theta$ corresponds to the average total variation distance between the two conditional distributions $f_{Y|X}$ and $g_{V|U}$ \citep[Remark 2]{hu_two-sample_2023}, and can be viewed as a measure of deviance from the covariate shift assumption.

As a consequence, testing $H_0$ can be cast as a problem of constructing confidence intervals for the parameter $\theta$.   To this end, the nuisance parameter $\gamma$ must also be estimated. Without additional assumptions, this parameter is naturally non-parametric so we may be interested in using non-parametric methods such as random forests or neural nets. Moreover, in a high-dimensional setting, we may also wish to use regularized estimators such as LASSO. However, if the test statistic is not orthogonal, the estimate will likely not have the correct type-1 control. Unfortunately, as shown in the following proposition, this is the case.  This proposition also points to the corresponding debiasing scheme.
\begin{proposition}\label{pro:not_orthogonal}
Let $H = \delta(x',y')$ and $K=\delta(u',v')$ be pointwise contaminants at some fixed points $(x',y')$ and $(u',v')$. Consider a one-dimensional sub-model: $\{(F_\epsilon, G_\epsilon):\epsilon\in[0,1]\}$ with $F_\epsilon=(1-\epsilon)F+\epsilon H$, $G_\epsilon=(1-\epsilon)G+\epsilon K$. The pathwise derivative of $\U\varphi(X,Y,U,V,\gamma_\epsilon)$ at $\epsilon=0$ is
    \begin{equation*}
        \frac{d}{d\epsilon}\U [\gamma_\epsilon(X)a(X,Y,U,V)] \big|_{\epsilon = 0} = \alpha(u')-\gamma(x')\alpha(x')\,,
    \end{equation*}
     where 
     \begin{equation}\label{eq:alpha}
     \alpha(X)=E_{U,V\sim G,Y \sim f(y|x)} a(X,Y,U,V)\,.\end{equation}
As a result, the functional $\U\varphi(X,Y,U,V,\gamma)$ is not Neyman orthogonal with respect to the nuisance parameter $\gamma$.
\end{proposition}

\subsection{Debiased test statistic}
To orthogonalize the kernel, we need to find the U-statistic influence function. When $H,K$ are point-mass contaminants as above, notice that, by Proposition \ref{pro:not_orthogonal},
\begin{align*}
        &\frac{d}{d\epsilon}\U[\gamma_\epsilon(X)a(X,Y,U,V)] \big|_{\epsilon = 0}\\
         =& \alpha(u')-\gamma(x')\alpha(x')\\
        = &\int \left[\alpha(u)-\gamma(x)\alpha(x)\right] [dF(x)\delta(u') + \delta(x')dG(u)].
    \end{align*}
    where $\alpha(\cdot)$ is defined as in \eqref{eq:alpha}.

Thus, when $H,K$ are point mass contaminants, the function $\phi(X,Y,U,V,\gamma,\alpha) = \alpha(U)-\gamma(X)\alpha(X)$ satisfies pathwise differentiability, which makes it a candidate for the influence function.
To prove that this is the influence function, we need to show the zero-mean property \eqref{mean zero} and that pathwise differentiability \eqref{pathwise diff} holds for arbitrary paths. 

\begin{theorem}
    \label{cov shift influence function}
    The function  $\phi(X,Y,U,V,\gamma,\alpha) = \alpha(U)-\gamma(X)\alpha(X)$ is a U-statistic influence function for the two-sample U-statistic $\varphi(X,Y,U,V,\gamma) = \gamma(X)a(X,Y,U,V)$.
\end{theorem}

\Cref{cov shift influence function} is proved in \Cref{app:proof_sec4}. As a sanity check that we have the correct influence function, we can verify if the bias is indeed second-order in the estimation error of the nuisance parameters $\gamma$ and $\alpha$. The bias can be computed as
\begin{equation*}
    \begin{split}
    &\U[\hat{\gamma}(X)a(X,Y,U,V)+\hat{\alpha}(U)-\hat{\gamma}(X)\hat{\alpha}(X)] -\U[\gamma(X)a(X,Y,U,V)] \\
    &= \U[(\hat{\gamma}(X)-\gamma(X))a(X,Y,U,V)] +E[\gamma(X)\hat{\alpha}(X)-\hat{\gamma}(X)\hat{\alpha}(X)]\\
    &= E[(\hat{\gamma}(X)-\gamma(X))\alpha(X)] + E[\gamma(X)\hat{\alpha}(X)-\hat{\gamma}(X)\hat{\alpha}(X)]\\
    &= E[(\hat{\gamma}(X)-\gamma(X))(\alpha(X)-\hat{\alpha}(X))]\,,
    \end{split}
\end{equation*}
where the third equality follows from taking expectation with respect to $Y,U$, and $V$. This calculation also highlights the double robustness of this kernel. Our estimate will be unbiased as long as one of $\gamma$ or $\alpha$ is estimated well-enough. In particular, the correct specification of $\alpha$ will cancel out the bias induced by the estimation of $\gamma$. 

Knowing the influence function, we can debias $\varphi$ to get a debiased U-statistic kernel
\begin{equation}\label{eq:debiased-psi-covshift}
    \psi(X,Y,U,V) = \gamma(X)a(X,Y,U,V) + \alpha(U)-\gamma(X)\alpha(X)\,.
\end{equation}

\subsection{Asymptotics}
Now we check the assumptions for asymptotic normality. Recall that we required assumptions \ref{as: sample ratio}, \ref{as: consistency}, and \ref{as: bias} for the debiased U-statistic to be asymptotically normal. We provide sufficient conditions for the debiased test statistic $\psi(X,Y,U,V)$ given in \eqref{eq:debiased-psi-covshift} to satisfy these assumptions and thus be asymptotically normal.

The first batch of conditions are about the consistency of nuisance parameters.
\begin{condition}
\label{cond 1}
For every fold $I_{st}$, we have
\begin{enumerate}
    \item [(a)] $\|\hat{\gamma}(X) - \gamma(X)\| = o_p(1)$
    \item [(b)] $\|\hat{\alpha}(U) - \alpha(U)\| = o_p(1)$
    \item [(c)] $\|\gamma(X)\hat{\alpha}(X) - \gamma(X)\alpha(X)\| = o_p(1).$
\end{enumerate}
\end{condition}
Parts (a) and (b) simply require consistency of $\hat\gamma$ and $\hat\alpha$ in $L_2$-norm.  Part (c) can be implied by part (b) if $\|\gamma\|_\infty<\infty$.


\begin{condition}
\label{cond 3}
    For each fold $(s,t)$, we have that $\|\hat{\gamma} - \gamma\|\cdot\|\hat{\alpha} - \alpha\| = o_p\left(\frac{1}{\sqrt{m+n}}\right)$.
\end{condition}
Condition \ref{cond 3} is at the core of the debiasing technique.  It allows us to estimate both nuisance functions $\gamma$ and $\alpha$ at a rate strictly worse than the parametric rate, as long as their product is dominated by the parametric rate.

\begin{theorem}
\label{thm test normal}
    Under Assumption \ref{as: sample ratio}, Conditions \ref{cond 1} and \ref{cond 3}, the debiased U-statistic 
    $$\U_{mn}\psi(X,Y,U,V) =\U_{mn}\left[ \gamma(X)a(X,Y,U,V) + \alpha(U)-\gamma(X)\alpha(X)\right]$$ is asymptotically normal. 
\end{theorem}

To implement the test, we need an estimate for the asymptotic variance. The asymptotic variance we need to estimate is the variance of    $\psi(X,Y,U,V,\gamma,\alpha)$ in \eqref{eq:debiased-psi-covshift}.
Standard U-statistic theory implies that the variance is 
\begin{equation*}
    \frac{\var(E[\psi(X,Y,U,V,\gamma,\alpha)|X,Y])}{\lambda} + \frac{\var(E[\psi(X,Y,U,V,\gamma,\alpha)|U,V])}{1-\lambda},
\end{equation*}
where $\lambda = \lim_{m.n\rightarrow \infty} \frac{m}{m+n}$ as specified in Assumption \ref{as: sample ratio}.

Following the discussion in section 3, we estimate the variance as
\begin{align}
        \hat{\sigma}^2 &=\frac{\frac{1}{m}\sum_{s=1}^S\sum_{i \in C_s}\left(\frac{1}{n}\sum_{t=1}^T\sum_{j \in D_t}\psi(x_i,y_i,u_j,v_j,\gamma^{-st},\alpha^{-st})-0.5\right)^2}{\lambda}\nonumber\\
        &+ \frac{\frac{1}{n}\sum_{t=1}^T\sum_{j\in D_t}\left(\frac{1}{m}\sum_{s=1}^S\sum_{i \in C_s}\psi(x_i,y_i,u_j,v_j,\gamma^{-st},\alpha^{-st})-0.5\right)^2}{1-\lambda}\,,\label{var est}
\end{align}
where, $\gamma^{-st}$ and $\alpha^{-st}$ means the estimate of $\gamma$ and $\alpha$ using points outside of the fold $C_s \times D_t$.

\subsection{Implementation}

We outline the implementation of the debiased covariate shift test. For convenience of notation, we index our samples by $\{(X_i,Y_i)\}_{i=1}^{2m}$ and $\{U_j,V_j\}_{j=1}^{2n}$.

\begin{enumerate}
    \item Use data $\{(X_i,Y_i)\}_{i=m}^{2m}$ and $\{U_j,V_j\}_{j=n}^{2n}$ to get an estimate $\hat{a}$ of $a$
    \item Randomly split $\{(X_i,Y_i)\}_{i=1}^{m}$ into folds $C_1,...,C_S$ and $\{U_j,V_j\}_{j=1}^{n}$ into folds $D_1,...,D_T$.
    \item For each $(s,t)$, estimate $\gamma_{st}$ and $\alpha_{st}$ using $C_s\times D_t$ and compute $\theta_{st}$ as in Equation \eqref{eq:theta_st}.
    \item Aggregate the on fold estimates to get $\hat{\theta}$ as in \eqref{eq:theta_crossfitting}.
    \item Compute $\hat{\sigma}^2$ using \eqref{var est}.
    \item Compute the standardized test statistic $\hat{T} = \frac{(.5-\hat{\theta})}{\hat{\sigma}/\sqrt{m+n}}$.
    \item Reject the null if $\hat{T}> \Phi^{-1}(1-\alpha)$.
\end{enumerate}

\paragraph{Degenerate Test Statistic}
Under the null, the nuisance parameter $s$, the conditional density ratio, is identically $1$. Taking $s=1$ results in a degenerate U-statistic kernel to which the above theory does not apply. In practice, we set a side part of the data data to get an estimate $\hat{s}$. We treat this $\hat{s}$ as fixed and then apply the double/debiased framework. If the fitted function $\hat s$ has a discrete distribution, we can use a random tie-breaking in the definition of $a$. In particular, let $(\zeta_i,~1\le i\le m)$ and $(\zeta_j'~,1\le j\le n)$ be independent $U(0,1)$ random variables that are also independent of the data.  Define 
$$
a(X_i,Y_i,U_j,V_j)=\mathds{1}(\hat s(X_i,Y_i)<\hat s(U_j,V_j))+\mathds{1}(\zeta_i<\zeta_j')\mathds{1}(\hat s(X_i,Y_i)=\hat s(U_j,V_j))\,. 
$$
Then the U-statistic kernel is guaranteed to be non-degenerate.

\paragraph{Estimation of Density Ratios}
The nuisance parameters $\gamma$ and $s$ both involve the estimation of density ratios. There is plenty of literature on methods for estimating density ratios \citep{sugiyama_density_nodate}. We recommend using a classification method for density ratio estimation. There are two major benefits to using this method. First, these methods do not involve the estimation of densities directly. Second, this method only requires a classification method of which there are many machine learning methods available. 

Let's illustrate the classification method for marginal density ratio $\gamma = \frac{g(X)}{f(X)}$. Suppose we have observations $\{X_i\}_{i=1}^m$ and $\{U_i\}_{i=1}^n$. Assign the $X$ sample to class $0$ and the $U$ sample to class $1$. Using any classification method, train a model to predict $\eta(x):= P(1|X=x)$. Then we can estimate $\gamma$ by
\begin{equation*}
    \hat{\gamma}(X) = \frac{n}{m}\frac{1-\hat{\eta}(X)}{\hat{\eta}(X)}.
\end{equation*}

\paragraph{Estimation of $\alpha$}
The final nuisance parameter we need to estimate is $\alpha$. We can treat this as a regression problem in the following sense. Consider the function
\begin{equation*}
    a^*(X,Y) = \frac{1}{n}\sum_{j=1}^n \hat{a}(X,Y,U_j,V_j).
\end{equation*}
We can then get an estimate of $\alpha$ by regressing $a^*$ on $X$. This method allows the use of any regression method for estimating $\alpha$.


\section{Simulations}




We evaluate the finite sample performance of our method in three different settings. The first setting is the simplest where the nuisance parameter comes from a low-dimensional parametric model.

We also examine the effectiveness of our test when the nuisance parameter is more complex and requires flexible estimators. These settings involve a high-dimensional parametric model which requires regularized estimators and a non-parametric example using NASA airfoil data \citep{misc_airfoil_self-noise_291}. These settings are similar to the ones studied in \citet{hu_two-sample_2023}.

\subsection{The low-dimensional setting}
In the low-dimensional parametric setting, we sample covariates $X \sim N(0,I)$ and $U \sim N(\mu,I)$ where $\mu=(1,-1,1,-1,0)^T$ and $I$ is the identity matrix. The response is given by $y = \beta^T x + \epsilon$ and $v = \alpha + \beta^T u + \epsilon$ where $\epsilon \sim N(0,1)$, $\beta = (1,1,1,1,1)$ and $\alpha=0$ in the null case and $\alpha=0.5$ in the alternate case.

To estimate all nuisance parameters we use logistic regression. This is a parametrically correct model for the density ratios and performs well enough for estimating $\alpha$ to get correct inference. All results are calculated empirically with 500 trials. All tests are conducted with nominal type-1 control at $0.05$.

\paragraph{Density ratio cutoff}
Following the data-generating procedure in \citet{hu_two-sample_2023}, we will cutoff the range of the density ratios. We remove sample points where the true marginal density ratio lies outside of the interval $[1/50,50]$. This cutoff allows for more stable density ratio estimation which is required for correct inference in our method. 

\paragraph{Sample split ratio}
To implement this method we need to make a sample split corresponding to the nuisance parameter $s$, the full density ratio, and $\gamma$ and $\alpha$ the nuisance parameters used in cross-fitting. For better performance, we will use 2/3 of the data for $s$ and 1/3 of the available data for the cross-fitting procedure. In the simulations that follows, $n$ denotes the sample size used in the cross-fitting procedure while $2n$ samples will be used for estimating $s$.

\begin{table}[H]
\centering
\caption{Simulation results for the debiased and plugin test statistics in the low-dimensional setting. The debiased test statistic is unbiased (up to two decimal places) for all sample sizes, and the type-1 control is around the $0.05$ nominal level in all cases. As the sample size increases, the power increases to $1$. The results for the plugin test statistic show correct type-1 error and power going to $1$. However, there is a slight upward bias leading to a conservative test and worse empirical power compared to the debiased test.}
\begin{tabular}{l|lll|lll}
\multirow{2}{*}{n} & \multicolumn{3}{l|}{Debiased} & \multicolumn{3}{l}{Plugin}   \\
                   & bias  & type-1 error  & power & bias  & type-1 error & power \\ \hline
250                & 0.00  & 0.034         & 0.64  & 0.11  & 0.020        & 0.32  \\
500                & 0.00  & 0.060         & 0.84  & 0.043 & 0.014        & 0.50  \\
1000               & 0.00  & 0.062         & 0.99  & 0.023 & 0.012        & 0.80  \\
2000               & 0.00  & 0.044         & 1.00  & 0.011 & 0.016        & 0.98 
\end{tabular}
\label{tb: low-dim sim}
\end{table}

Table \ref{tb: low-dim sim} show the simulation results for both the debiased and plug-in test statistics in the low-dimensional setting. In the debiased setting, the empirical bias, the difference between empirical mean of the estimates over the 500 trials and the true value $0.5$, is $0$ up to two decimal places and is close to the correct nominal type-1 error of $0.05$ with power going to $1$. For the plug-in estimator, the test is slightly biased leading to conservative results. We see that the debiased test leads to a significant improvement in the alternate setting.

\subsection{The high-dimensional setting}
In the high-dimensional parametric setting, we sample covariates $X \sim N(0,I)$ and $U \sim N(\mu,I)$ where $\mu=(1,-1,1,-1,0,...,0)^T$ and $I$ is the identity matrix. Here $\mu$ is a $500$-dimensional vector. The response is given by $y = \beta^T x + \epsilon$ and $v = \beta_0 + \beta^T u + \epsilon$ where $\epsilon \sim N(0,1)$, $\beta = (1,1,1,1,1,0,...,0)$ and $\beta_0=0$ in the null case and $\beta_0=0.25$ in the alternate case.

We again conduct all tests at $0.05$ level type-1 control. All empirical bias, type-1 error and power are calculated over $500$ trials. As we are in a high-dimensional setting, we will use LASSO to estimate the nuisance parameters $\gamma$ and $\alpha$. In this case, LASSO is the correct model for $\gamma$ but in general does not converge at the parametric rate, which is why debiasing is required.

To estimate $s$, we will use stability selection \citep{meinshausen_stability_2010} for variable selection and then logistic regression. The reason for this is that $\alpha$ depends on $s$ through an indicator function. LASSO for $s$ might still return many non-zero but small coefficients. However, for an indicator function, magnitude does not matter. In this case, estimating $\alpha$ may be a non-sparse high-dimensional problem which is highly difficult. Applying a variable selection method first such as using stability selection will force estimating $\alpha$ to be intrinsically low-dimensional, leading to better results. 

\begin{figure}[H]
\centering
\begin{subfigure}{.5\textwidth}
  \centering
  \includegraphics[width=8cm]{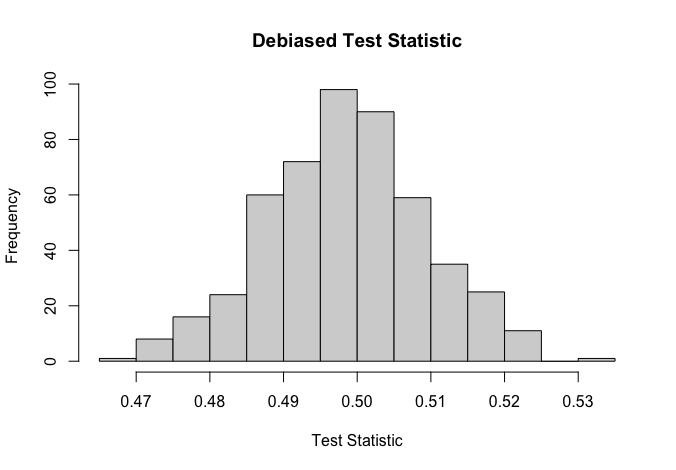}
  \caption{Histogram of debiased test statistic.}
  \label{fig:debiased}
\end{subfigure}%
\begin{subfigure}{.5\textwidth}
  \centering
  \includegraphics[width=8cm]{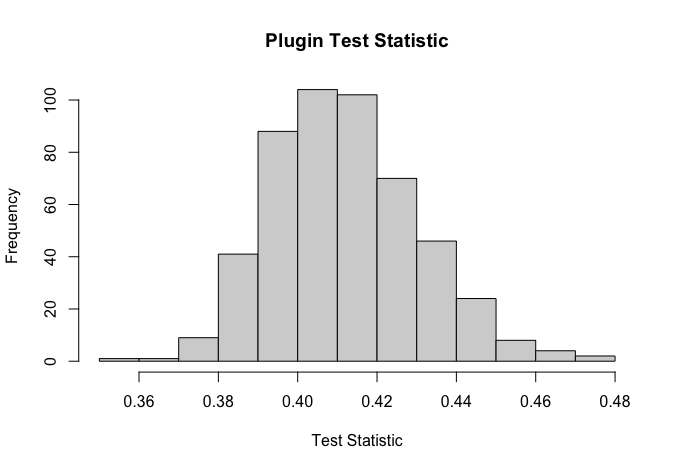}
  \caption{Histogram of plug-in test statistic.}
  \label{fig:plugin}
\end{subfigure}
\caption{Illustration of the debiasing effect in the high-dimensional setting. The empirical distribution of the debiased test statistic on the left is approximately normal and centered at the correct value of $0.5$. Meanwhile, the empirical distribution is centered at a value significantly less than $0.5$.}
\label{fig:test}
\end{figure}

\begin{table}[H]
\centering
\caption{\label{tb: high-dim sim}Simulation results for both the debiased and plugin test statistics in the high-dimensional setting. The debiased test has significantly less bias and much better type-1 control than the plugin test.}
\begin{tabular}{l|lll|lll}
\multirow{2}{*}{n} & \multicolumn{3}{l|}{Debiased}  & \multicolumn{3}{l}{Plugin}    \\
                   & bias    & type-1 error & power & bias   & type-1 error & power \\ \hline
250                & -0.0068 & 0.15         & 0.61  & -0.09  & 0.61         & 0.73  \\
500                & -0.0036 & 0.10         & 0.82  & -0.11  & 0.40         & 0.93  \\
1000               & -0.0030 & 0.13         & 0.95  & -0.10  & 0.31         & 1.00  \\
2000               & -0.0015 & 0.076        & 1.00  & -0.088 & 0.25         & 1.00 
\end{tabular}
\end{table}

\Cref{tb: high-dim sim} shows the results of the debiased and plugin tests for the high dimensional test. The plugin estimator has a significantly larger empirical bias as seen in Figure \ref{fig:test}. As a result, the debiased test has significantly better type-1 control compared to the plugin test. 

\subsection{Airfoil Data}
We demonstrate the effectiveness of our method on real data through the airfoil dataset \citep{misc_airfoil_self-noise_291} previously studied in \citet{hu_two-sample_2023} and \citet{tibshirani_conformal_2020}. This dataset collected by NASA studied the sound pressure of various airfoils. The dataset consists of 5 covariates, log frequency, angle of attack, chord length, free-stream
velocity, suction side log displacement thickness, and a response variable scaled sound pressure. 

The airfoil dataset does not naturally have distinct samples. We artificially split the $n=1503$ observations into two samples in the following ways:
\begin{enumerate}
    \item Random partition
    \item Exponential tilting
    \item Partition along velocity covariate
    \item Partition along the response 
\end{enumerate}

In the first three partitions, the two samples satisfy the covariate shift assumption. For the first partition, both samples come from the same distribution, while in the second and third partitions, there is a non-trivial covariate shift. In the last partition, the covariate shift assumption is not satisfied. 

\paragraph{Random Sampling}

In this setting, the observations are randomly split into two groups to create a two sample scenario in the null case. The covariate distributions in the two groups are identical so the covariate shift assumption is satisfied. We tried this method using linear logistic regression and SuperLearner. Superlearner is a stacking method that combines estimates from multiple ML methods into a single estimate by weighting, where the weights are chosen by cross-validation \citep{van_der_laan_super_2007}. We use SuperLearner with linear regression, SVM, random forests, and xgboost to estimate the nuisance parameters. The random sampling was repeated $500$ times and the mean estimate and rejection proportions are recorded in Table \ref{tb: random sampling}. Both methods on average correctly estimate the test statistic to be $0.5$ and the rejection proportion is close to the desired $0.05$ level.

\begin{table}[H]
\centering
\caption{Random splitting. Mean estimates of the test statistic and rejection proportion over 500 repetitions using linear logistic regression and SuperLearner to estimate nuisance parameters. Both methods correctly estimate the test statistic and achieve the correct $0.05$ type-1 error control}

\begin{tabular}{l|ll}
             & Mean    & Rejection Proportion \\ \hline
Linear       & 0.500 & 0.048                \\
SuperLearner & 0.500 & 0.048               
\end{tabular}
\label{tb: random sampling}
\end{table}

\paragraph{Exponential Tilting}
To create a non-trivial covariate shift setting, we use exponential tilting to form two samples. We first randomly split the data into two samples. The first sample is kept as it is. In the second sample, we resample with replacement with distribution $\exp(x^T\alpha)$ where $\alpha=(-1,0,0,0,1)$. This is the same exponential tilting setting as in \cite{hu_two-sample_2023}. Table \ref{tb: exp tilt} shows the results where the nuisance parameters are estimated with both logistic regression and superlearner.  We again see that we get an approximate $0.05$ empirical type-1 error and the average of the test statistics is close to $0.5$ as expected.

\begin{table}[H]
\caption{Exponentially tilted sampling splitting. Mean estimates of the test statistic and rejection proportion over 500 repetitions using linear logistic regression and SuperLearner to estimate nuisance parameters. In both cases, the mean estimate is close to the correct value of $0.5$ and the rejection proportion is close to the correct nominal type-1 error of $0.05$.}
\centering
\begin{tabular}{l|ll}
             & Mean  & Rejection Proportion \\ \hline
Linear       & 0.498 & 0.056                \\
SuperLearner & 0.496 & 0.056               
\end{tabular}
\label{tb: exp tilt}
\end{table}

\paragraph{Velocity Partition}
We also form a non-trivial covariate shift setting potentially in the alternative case by splitting the data into two samples based on the velocity covariate. This is done by splitting exactly at the median and then randomly flipping five percent of the observations to the other sample. In this setting, we only get a single realization of the data. Following the setting in \citet{hu_two-sample_2023}, we aggregate p-values by using the twice the median p-value \citep{diciccio_exact_2020} over the randomization of the flips and sample splits. Table \ref{tb: velocity partition} shows the results where both nuisance parameters are estimated using logistic regression, where $\gamma$ is estimated using logistic regression and $\alpha$ is estimated using SuperLearner, and where both are estimated using super learner. We see that only when both are estimated using SuperLearner does the test correctly not reject the null. Previous work has shown that non-parametric estimators of the nuisance parameters perform better than linear methods in this setting \citep{hu_two-sample_2023}.

\begin{table}[H]
\caption{Velocity partition. In this setting the covariate shift assumption is satisfied as the partition is along a covariate. Mean estimates of the test statistic, median standard errors and rejection proportion using linear logistic regression and SuperLearner to estimate nuisance parameters. When Superlearner is used for both nuisance parameters, the test correctly does not reject the null.}
\centering
\begin{tabular}{l|lll}
                                     & Mean & Median Standard Error & p-value \\ \hline
Linear                               & 0.38 & 0.019                 & 0       \\
Linear $\hat\gamma$, SL $\hat\alpha$ & 0.40 & 0.012                 & 0       \\
SL                                   & 0.52 & 0.056                 & 1   
\end{tabular}
\label{tb: velocity partition}
\end{table}

\paragraph{Sound Partition}
In this setting, we divide into two samples based on the sound variable in the same way as the velocity partition. However, as sound is the response, this will give a two-sample split that does not satisfy the covariate shift assumption. We tried the test using the same settings as the velocity partition. In this case, all tests correctly reject the null.

\begin{table}[H]
\caption{Sound partition. In this setting, the covariate shift assumption is not satisfied as the partition is along the response. Mean estimates of the test statistic, median standard errors, and rejection proportion using linear logistic regression and SuperLearner to estimate nuisance parameters. In all cases, the test correctly rejects the null.}
\centering
\begin{tabular}{l|lll}
                                     & Mean  & Median Standard Error & p-value \\ \hline
Linear                               & 0.018 & 0.044                 & 0       \\
Linear $\hat\gamma$, SL $\hat\alpha$ & 0.12  & 0.033                 & 0       \\
SL                                   & 0.27  & 0.067                 & 0.0005 
\end{tabular}
\end{table}


\section{Discussion}

This work extends the double/debiased machine learning framework to two sample problems allowing flexible analysis of two-sample functionals with nuisance parameters. In general, the estimation of these functionals with asymptotic normality is possible under (strong) assumptions that nuisance parameters are parametric. However, in many interesting settings, we are interested in non-parametric or regularized estimation of the nuisance parameter. These cases involve using estimators for the nuisance parameter with slower than $\sqrt{n}$ convergence rates. Typically, this will result in an estimate which is not asymptotically normal. Under the debiased/double machine learning framework, we alter the functional by adding a mean-zero function so that the target parameter is still the same but the impact of the estimation error for the nuisance parameters is at most second-order. This allows for slower than $\sqrt{n}$ convergence of nuisance parameter estimates while still attaining asymptotic normality.

As an application of this method and theory, we look at conformal-based conditional distribution testing, specifically a two-sample test for covariate shift. In this setting, we consider the conformal rank-sum test statistic developed in \cite{hu_two-sample_2023}, where the validity depends on the estimation of the marginal density ratios of the two samples. In many settings, the density ratio may be nonparametric or high-dimensional requiring the use of black-box ML or regularized methods which have slower convergence rates. Naturally, this fits into the double/debiased setting. In our work, we show how to debias the test statistic by augmenting the test statistic with the two-sample U-statistic influence function. Using this debiased test statistic along with cross-fitting leads to a valid test statistic with milder conditions on the accuracy of the nuisance parameter estimation.

\medskip
\nocite{*}
\bibliographystyle{plainnat}
\bibliography{bibliography}

\begin{thebibliography}{24}
\providecommand{\natexlab}[1]{#1}
\providecommand{\url}[1]{\texttt{#1}}
\expandafter\ifx\csname urlstyle\endcsname\relax
  \providecommand{\doi}[1]{doi: #1}\else
  \providecommand{\doi}{doi: \begingroup \urlstyle{rm}\Url}\fi

\bibitem[Bickel and Ritov(1988)]{bickel_estimating_1988}
P.~J. Bickel and Y.~Ritov.
\newblock Estimating {Integrated} {Squared} {Density} {Derivatives}: {Sharp} {Best} {Order} of {Convergence} {Estimates}.
\newblock \emph{Sankhyā: The Indian Journal of Statistics, Series A (1961-2002)}, 50\penalty0 (3):\penalty0 381--393, 1988.
\newblock ISSN 0581-572X.
\newblock Publisher: Springer.

\bibitem[Bolthausen et~al.(2002)Bolthausen, Perkins, and {van der Vaart}]{bolthausen_lectures_2002}
Erwin Bolthausen, Edwin Perkins, and Aad {van der Vaart}.
\newblock \emph{Lectures on probability theory and statistics}.
\newblock Number 1781 in Lecture notes in mathematics. Springer, Berlin Heidelberg, 2002.
\newblock ISBN 978-3-540-43736-9.

\bibitem[Brooks et~al.(2014)Brooks, Pope, and Marcolini]{misc_airfoil_self-noise_291}
Thomas Brooks, D.~Pope, and Michael Marcolini.
\newblock {Airfoil Self-Noise}.
\newblock UCI Machine Learning Repository, 2014.
\newblock {DOI}: https://doi.org/10.24432/C5VW2C.

\bibitem[Chernozhukov et~al.(2018)Chernozhukov, Chetverikov, Demirer, Duflo, Hansen, Newey, and Robins]{chernozhukov_doubledebiased_2018}
Victor Chernozhukov, Denis Chetverikov, Mert Demirer, Esther Duflo, Christian Hansen, Whitney Newey, and James Robins.
\newblock Double/debiased machine learning for treatment and structural parameters.
\newblock \emph{The Econometrics Journal}, 21\penalty0 (1):\penalty0 C1--C68, February 2018.
\newblock ISSN 1368-4221, 1368-423X.
\newblock URL \url{https://doi.org/10.1111/ectj.12097}.

\bibitem[Chernozhukov et~al.(2022)Chernozhukov, Escanciano, Ichimura, Newey, and Robins]{chernozhukov_locally_2022}
Victor Chernozhukov, Juan~Carlos Escanciano, Hidehiko Ichimura, Whitney~K. Newey, and James~M. Robins.
\newblock Locally {Robust} {Semiparametric} {Estimation}.
\newblock \emph{ECTA}, 90\penalty0 (4):\penalty0 1501--1535, 2022.
\newblock ISSN 0012-9682.
\newblock URL \url{https://doi.org/10.3982/ECTA16294}.

\bibitem[DiCiccio et~al.(2020)DiCiccio, DiCiccio, and Romano]{diciccio_exact_2020}
Cyrus~J. DiCiccio, Thomas~J. DiCiccio, and Joseph~P. Romano.
\newblock Exact tests via multiple data splitting.
\newblock \emph{Statistics \& Probability Letters}, 166:\penalty0 108865, November 2020.
\newblock ISSN 01677152.
\newblock URL \url{https://doi.org/10.1016/j.spl.2020.108865}.

\bibitem[Escanciano and Terschuur(2023)]{escanciano_debiased_2023}
Juan~Carlos Escanciano and Joël~Robert Terschuur.
\newblock Debiased {Semiparametric} {U}-{Statistics}: {Machine} {Learning} {Inference} on {Inequality} of {Opportunity}, January 2023.
\newblock URL \url{http://arxiv.org/abs/2206.05235}.

\bibitem[Hasminskii and Ibragimov(1986)]{hasminskii_asymptotically_1986}
R.~Z. Hasminskii and I.~A. Ibragimov.
\newblock Asymptotically efficient nonparametric estimation of functionals of a spectral density function.
\newblock \emph{Probab. Th. Rel. Fields}, 73\penalty0 (3):\penalty0 447--461, September 1986.
\newblock ISSN 0178-8051, 1432-2064.
\newblock URL \url{https://doi.org/10.1007/BF00776242}.

\bibitem[Hu and Lei(2023)]{hu_two-sample_2023}
Xiaoyu Hu and Jing Lei.
\newblock A {Two}-{Sample} {Conditional} {Distribution} {Test} {Using} {Conformal} {Prediction} and {Weighted} {Rank} {Sum}.
\newblock \emph{Journal of the American Statistical Association}, pages 1--19, March 2023.
\newblock ISSN 0162-1459, 1537-274X.
\newblock URL \url{https://doi.org/10.1080/01621459.2023.2177165}.

\bibitem[Kennedy(2023)]{kennedy_semiparametric_2023}
Edward~H. Kennedy.
\newblock Semiparametric doubly robust targeted double machine learning: a review, January 2023.
\newblock URL \url{http://arxiv.org/abs/2203.06469}.

\bibitem[Meinshausen and Bühlmann(2010)]{meinshausen_stability_2010}
Nicolai Meinshausen and Peter Bühlmann.
\newblock Stability {Selection}.
\newblock \emph{Journal of the Royal Statistical Society Series B: Statistical Methodology}, 72\penalty0 (4):\penalty0 417--473, September 2010.
\newblock ISSN 1369-7412, 1467-9868.
\newblock URL \url{https://doi.org/10.1111/j.1467-9868.2010.00740.x}.

\bibitem[Newey and Robins(2018)]{newey_cross-fitting_2018}
Whitney~K. Newey and James~R. Robins.
\newblock Cross-{Fitting} and {Fast} {Remainder} {Rates} for {Semiparametric} {Estimation}, January 2018.
\newblock URL \url{http://arxiv.org/abs/1801.09138}.

\bibitem[Pan and Yang(2010)]{pan_survey_2010}
Sinno~Jialin Pan and Qiang Yang.
\newblock A {Survey} on {Transfer} {Learning}.
\newblock \emph{IEEE Trans. Knowl. Data Eng.}, 22\penalty0 (10):\penalty0 1345--1359, October 2010.
\newblock ISSN 1041-4347.
\newblock URL \url{https://doi.org/10.1109/TKDE.2009.191}.

\bibitem[Peters et~al.(2016)Peters, Bühlmann, and Meinshausen]{peters_causal_2016}
Jonas Peters, Peter Bühlmann, and Nicolai Meinshausen.
\newblock Causal inference using invariant prediction: identification and confidence intervals.
\newblock \emph{Journal of the Royal Statistical Society: Series B (Methodological)}, 78:\penalty0 947--1012, November 2016.
\newblock URL \url{https://doi.org/10.1111/rssb.12167}.

\bibitem[Pfanzagl and Wefelmeyer(1982)]{pfanzagl_contributions_1982}
J.~Pfanzagl and W.~Wefelmeyer.
\newblock \emph{Contributions to a general asymptotic statistical theory}.
\newblock Number~13 in Lecture notes in statistics. Springer-Verlag, New York, 1982.
\newblock ISBN 978-0-387-90776-5.

\bibitem[Qiu et~al.(2023)Qiu, Dobriban, and Tchetgen]{qiu_prediction_2023}
Hongxiang Qiu, Edgar Dobriban, and Eric~Tchetgen Tchetgen.
\newblock Prediction {Sets} {Adaptive} to {Unknown} {Covariate} {Shift}, June 2023.
\newblock URL \url{http://arxiv.org/abs/2203.06126}.

\bibitem[Rosenbaum and Rubin(1983)]{rosenbaum_central_1983}
Paul~R. Rosenbaum and Donald~B. Rubin.
\newblock The central role of the propensity score in observational studies for causal effects.
\newblock \emph{Biometrika}, 70\penalty0 (1):\penalty0 41--55, 1983.
\newblock ISSN 0006-3444, 1464-3510.
\newblock URL \url{https://doi.org/10.1093/biomet/70.1.41}.

\bibitem[Schick(1986)]{schick_asymptotically_1986}
Anton Schick.
\newblock On {Asymptotically} {Efficient} {Estimation} in {Semiparametric} {Models}.
\newblock \emph{Ann. Statist.}, 14\penalty0 (3), September 1986.
\newblock ISSN 0090-5364.
\newblock URL \url{https://doi.org/10.1214/aos/1176350055}.

\bibitem[Sugiyama and Kawanabe(2012)]{sugiyama_machine_2012}
Masashi Sugiyama and Motoaki Kawanabe.
\newblock \emph{Machine learning in non-stationary environments: introduction to covariate shift adaptation}.
\newblock Adaptive computation and machine learning. MIT Press, Cambridge, Mass, 2012.
\newblock ISBN 978-0-262-01709-1.
\newblock OCLC: ocn752909553.

\bibitem[Sugiyama et~al.()Sugiyama, Suzuki, and Kanamori]{sugiyama_density_nodate}
Masashi Sugiyama, Taiji Suzuki, and Takafumi Kanamori.
\newblock Density {Ratio} {Estimation}: {A} {Comprehensive} {Review}.

\bibitem[Tibshirani et~al.(2019)Tibshirani, Foygel~Barber, Candes, and Ramdas]{tibshirani_conformal_2020}
Ryan~J Tibshirani, Rina Foygel~Barber, Emmanuel Candes, and Aaditya Ramdas.
\newblock Conformal prediction under covariate shift.
\newblock In H.~Wallach, H.~Larochelle, A.~Beygelzimer, F.~d\textquotesingle Alch\'{e}-Buc, E.~Fox, and R.~Garnett, editors, \emph{Advances in Neural Information Processing Systems}, volume~32. Curran Associates, Inc., 2019.

\bibitem[V.~Mises(1947)]{v_mises_asymptotic_1947}
R.~V.~Mises.
\newblock On the {Asymptotic} {Distribution} of {Differentiable} {Statistical} {Functions}.
\newblock \emph{Ann. Math. Statist.}, 18\penalty0 (3):\penalty0 309--348, September 1947.
\newblock ISSN 0003-4851.
\newblock URL \url{https://doi.org/10.1214/aoms/1177730385}.

\bibitem[Vaart(1998)]{vaart_asymptotic_1998}
A.~W. Van~Der Vaart.
\newblock \emph{Asymptotic {Statistics}}.
\newblock Cambridge University Press, 1 edition, October 1998.
\newblock ISBN 978-0-511-80225-6 978-0-521-49603-2 978-0-521-78450-4.
\newblock URL \url{https://doi.org/10.1017/CBO9780511802256}.

\bibitem[Van Der~Laan et~al.(2007)Van Der~Laan, Polley, and Hubbard]{van_der_laan_super_2007}
Mark~J. Van Der~Laan, Eric~C Polley, and Alan~E. Hubbard.
\newblock Super {Learner}.
\newblock \emph{Statistical Applications in Genetics and Molecular Biology}, 6\penalty0 (1), January 2007.
\newblock ISSN 1544-6115, 2194-6302.
\newblock URL \url{https://doi.org/10.2202/1544-6115.1309}.

\end{thebibliography}

\appendix

\section{Proofs from Sections 2 and 3}\label{app:proof_sec3}


\begin{proof}[Proof of \Cref{thm:neyman_ortho}]
    Consider differentiating both sides of \ref{mean zero} by $\epsilon$ and evaluating at zero. The derivative under the integral with the chain/product rule gives
\begin{equation*}
    \begin{split}
    &\frac{d}{d\epsilon} \phi(x,u,\gamma_\epsilon,\alpha_\epsilon) dF_\epsilon(x) dG_\epsilon(u)\big|_{\epsilon=0}\\
    &= \frac{d}{d\epsilon} \phi(x,u,\gamma_\epsilon,\alpha_\epsilon)\big|_{\epsilon=0} dF_0(x)dG_0(u) + \phi(x,u,\gamma_0,\alpha_0) \frac{d}{d\epsilon} dF_\epsilon(x)dG_\epsilon(u)\big|_{\epsilon=0}.
    \end{split}
\end{equation*}

Using more differentiation rules, we can simplify
\begin{equation*}
    \begin{split}
        \frac{d}{d\epsilon} dF_\epsilon(x)dG_\epsilon(u)\big|_{\epsilon=0} &= dF_0(x)\frac{d}{d\epsilon}dG_\epsilon(u)\big|_{\epsilon=0} + \frac{d}{d\epsilon}dF_\epsilon(x)\big|_{\epsilon=0}dG_0(u)\\
        &= dF_0(x)(dK-dG_0)(u) + dG_0(u)(dH-dF_0)(x)\\
        &= dF_0(x)dK(u)+dH(x)dG_0(u)-2dF_0(x)dG_0(u).
    \end{split}
\end{equation*}

Plugging all computations back in the integral and using that
\begin{equation*}
    \int\int \phi(x,u,\gamma_0,\alpha_0) F_0(dx)G_0(du) = 0,
\end{equation*}
we get the equation
\begin{equation*}
    \frac{d}{d\epsilon} \U[\phi(X,U,\gamma_\epsilon,\alpha_\epsilon)]\big|_{\epsilon=0} +  \int\int \phi(x,u,\gamma_0,\alpha_0) [dF_0(x)dK(u)+dH(x)dG_0(u)]= 0.
\end{equation*}
Using condition \ref{pathwise diff} on the second term then gives 
\begin{equation*}
    \frac{d}{d\epsilon} \U[\phi(X,U,\gamma_\epsilon,\alpha_\epsilon)]\big|_{\epsilon=0} + \frac{d}{d\epsilon} \U[\varphi(X,U,\gamma_\epsilon)]\big|_{\epsilon=0}= 0 \,,  
\end{equation*}
which is the desired result
\begin{equation*}
    \frac{d}{d\epsilon}\U[\psi(X,U,\gamma_\epsilon,\alpha_\epsilon)]\big|_{\epsilon=0} = 0.\qedhere
\end{equation*}
\end{proof}

\begin{proof}[Proof of Lemma \ref{lm: asymptotic}]
Assumption \ref{as: bias} already gives that the third term is asymptotically negligible. We only need to show the same holds for term 2. 
That is, under assumptions \ref{as: consistency} and \ref{as: sample ratio}, 
\begin{equation*}
    \sum_{s \in [S]} \sum_{t \in [T]}\frac{N_s M_t}{mn} B_{st} = o_p\left(\frac{1}{\sqrt{m+n}}\right).
\end{equation*}
Consider a single fold indexed by $(s,t)$. We have that 
\begin{equation}
\label{iter expansion}
    P\left(\left|\frac{M_s N_t}{mn} B_{st}\right| \geq t\right) = E\left[P\left(\left|\frac{M_s N_t}{mn} B_{st}\right| \geq t \mid I_{st}^c\right)\right],
\end{equation}
where $I_{st}^c$ denotes the observations outside the fold $C_s\times D_t$. The inside term on the right side can be bounded using Chebyshev's bound. Note that by conditioning on $I_{st}^c$, the estimated nuisance parameters are fixed. We need to compute the variance
\begin{equation*}
    \begin{split}
        &\var\left(\frac{M_s N_t}{mn} B_{st} \mid I_{st}^c\right)\\
         = & \var\left(\frac{1}{mn}\sum_{(i,j)\in I_{st}} \psi(X_i,U_j,\hat{\gamma}_{-st},\hat{\alpha}_{-st}) - \psi(X_i,U_j,\gamma_0,\alpha_0)\mid I_{st}^c\right)\\
        \leq & \frac{1}{(mn)^2} \sum_{i,j\in I_{st}} \sum_{p,q\in I_{st}}\cov(\psi(X_i,U_j,\hat{\gamma}_{-st},\hat{\alpha}_{-st}) - \psi(X_i,U_j,\gamma_0,\alpha_0)\,,\\
        &~~\psi(X_p,U_q,\hat{\gamma}_{-st},\hat{\alpha}_{-st}) - \psi(X_p,U_q,\gamma_0,\alpha_0)\mid I_{st}^c)\\
        \leq & \frac{\var(\hat{\psi}_{st} - \psi|I_{st}^c)}{mn} + \frac{m-1}{mn}\var(\hat{\psi}_{st} - \psi|I_{st}^c) + \frac{n-1}{mn}\var(\hat{\psi}_{st} - \psi|I_{st}^c)\\
        \leq & \frac{\|\hat{\psi}_{st} - \psi\|^2}{mn} + \frac{m-1}{mn}\|\hat{\psi}_{st} - \psi\|^2 + \frac{n-1}{mn}\|\hat{\psi}_{st} - \psi\|^2.
    \end{split}
\end{equation*}
The fourth line follows by counting the number of terms in common between $X_i,U_j,X_p,U_q$, noting that when these are all distinct the covariance term is $0$ by independence, and $(X_i,U_j)$ term appears in $M_s+N_t$ non-zero covariance terms, and control the covariances by Cauchy-Schwartz. Plugging the above computation into \ref{iter expansion}, we have
\begin{equation}
    \label{term 2 chebyshev}
    P\left(\left|\frac{N_s M_t}{mn} B_{st}\right| \geq t\right) \leq \frac{1}{t^2}\left[\frac{\|\hat{\psi}_{st} - \psi\|^2}{mn} + \frac{m-1}{mn}\|\hat{\psi}_{st} - \psi\|^2 + \frac{n-1}{mn}\|\hat{\psi}_{st} - \psi\|^2\right].
\end{equation}

    Using equation \ref{term 2 chebyshev}, we see that
    \begin{equation*}
        \begin{split}
            &P\left(\left|\sqrt{m+n}\frac{N_s M_t}{mn} B_{st}\right| \geq t\big| I_{st}^c\right)\\
            &\leq \frac{1}{t^2}\left[\frac{(m+n)}{mn}\|\hat{\psi}_{st} - \psi\|^2 + \frac{(m+n)(m-1)}{mn}\|\hat{\psi}_{st} - \psi\|^2 + \frac{(m+n)(n-1)}{mn}\|\hat{\psi}_{st} - \psi\|^2\right].
        \end{split}
    \end{equation*}

    See that
    \begin{equation*}
        \frac{m+n}{mn} = \frac{1}{n} + \frac{1}{m} \rightarrow 0.
    \end{equation*}

    We also have that
    \begin{equation*}
        \lim \frac{(m+n)(m-1)}{mn} = \lim \frac{m+n}{n} < \infty
    \end{equation*}
    and 
    \begin{equation*}
        \lim\frac{(m+n)(n-1)}{mn} =\lim \frac{m+n}{m} < \infty. 
    \end{equation*}
    The result follows since $\|\hat{\psi}_{st} -\psi\|=o_p(1).$
\end{proof}

\section{Proofs from Section 4}\label{app:proof_sec4}


\begin{proof}[Proof of Proposition \ref{pro:not_orthogonal}]
    We consider  for the path where $H = \delta(x',y')$ and $K=\delta(u',v')$ are pointwise contaminants at some fixed points $(x',y')$ and $(u',v')$. The nuisance function is then
    \begin{equation*}
        \gamma_\epsilon(x) = \frac{(1-\epsilon)g(x)+\epsilon\delta(u')}{(1-\epsilon)f(x)+\epsilon\delta(x')}.
    \end{equation*}
    Using derivative rules, see that
    \begin{equation*}
        \frac{d}{d\epsilon}\gamma_\epsilon(x)\big|_{\epsilon = 0} = \frac{\delta(u')}{f(x)} - \gamma(x)\frac{\delta(x')}{f(x)}.
    \end{equation*}
    Then we have that
    \begin{equation*}
        \frac{d}{d\epsilon}\U [\gamma_\epsilon(X)a(X,Y,U,V)] \big|_{\epsilon = 0} = \U \left[\left(\frac{\delta(u')}{f(X)} - \gamma(X)\frac{\delta(x')}{f(X)}\right)a(X,Y,U,V)\right].
    \end{equation*}

    Take the terms separately. We first consider
    \begin{equation*}
        \begin{split}
            \U\left[\frac{\delta(u')}{f(X)}a(X,Y,U,V)\right] &= E_{X\sim f}\left[ \frac{\delta(u')}{f(X)} \alpha(X)\right]\\ &=
            \int \alpha(x)\frac{\delta(u')}{f(x)} f(x) dx \\ &=
            \alpha(u'),
        \end{split}
    \end{equation*}
    where $\alpha(X)=E_{U,V\sim G,Y \sim f(y|x)} a(X,Y,U,V)$ and the first equality comes from integrating out $Y,U,V.$ We handle the second term in a similar way.

    \begin{equation*}
        \begin{split}
            \U\left[\gamma(X)\frac{\delta(x')}{f(X)}a(X,Y,U,V)\right] &= E_{X\sim f}\left[ \gamma(X)\frac{\delta(x')}{f(X)} \alpha(X)\right]\\ &=
            \int \gamma(x)\alpha(x)\frac{\delta(x')}{f(x)} f(x) dx \\ &=
            \gamma(x')\alpha(x').            
        \end{split}
    \end{equation*}
    Combining this all together shows that 
    \begin{equation*}
        \frac{d}{d\epsilon}\U [\gamma_\epsilon(X)a(X,Y,U,V)] \big|_{\epsilon = 0} = \alpha(u')-\gamma(x')\alpha(x')\,. \qedhere
    \end{equation*}
\end{proof}

\begin{proof}[Proof of Theorem \ref{cov shift influence function}]
    We start by showing the zero-mean condition. We have a path ending at $(H,K).$ Let $f_\epsilon$ and $g_\epsilon$ denote the $X$ and $U$ marginals of $F_\epsilon$ and $G_\epsilon$.
    Then we compute
    \begin{equation*}
        \begin{split}
            &\int\int \alpha_\epsilon(u)-\gamma_\epsilon(x)\alpha_\epsilon(x) dF_\epsilon(x,y)dG_\epsilon(u,v)\\
            &= \int \alpha_\epsilon(u) g_\epsilon(u) du - \int \gamma_\epsilon(x)\alpha_\epsilon(x) f_\epsilon(x) dx\\
            &= \int \alpha_\epsilon(u)g_\epsilon(u) du - \int \alpha_\epsilon(x)g_\epsilon(x) dx = 0.
        \end{split}
    \end{equation*}

    Next, we need to check pathwise differentiability.
    Let $h,k$ denote the $X$ and $U$ marginals of $H$ and $K$. First, we compute the left side of condition \ref{pathwise diff}
    \begin{equation*}
        \frac{d}{d\epsilon}\U [\gamma_\epsilon(X)a(X,Y,U,V)] \big|_{\epsilon = 0} = \U\left[\left(\frac{k(X)}{f(X)} - \gamma(X)\frac{h(X)}{f(X)}\right)a(X,Y,U,V)\right].
    \end{equation*}
    Breaking these terms up again, we first compute
    \begin{equation*}
        \begin{split}
            \U\left[\frac{k(X)}{f(X)}a(X,Y,U,V)\right] &= \int \frac{k(x)}{f(x)}\alpha(x) f(x)d(x)\\
            &= \int \alpha(x)k(x)dx.
        \end{split}
    \end{equation*}
    Here the steps of integrating out $Y,U,V$ to get $\alpha$ are the same as in the point mass case.
    Then the other term is
    \begin{equation*}
    \begin{split}
        \U\left[\gamma(X)\frac{h(X)}{f(X)}a(X,Y,U,V)\right] &= \int \gamma(x) \frac{h(x)}{f(x)}\alpha(x) f(x)dx\\
        &= \int \gamma(x)\alpha(x) h(x) dx.
    \end{split}
    \end{equation*}
    Thus, 
    \begin{equation*}
         \frac{d}{d\epsilon}\U [\gamma_\epsilon(X)a(X,Y,U,V)] \big|_{\epsilon = 0} = \int \alpha(x)k(x)dx - \int \gamma(x)\alpha(x) h(x) dx.
    \end{equation*}

    Now we compute from the right side of condition \ref{pathwise diff}. That is, we want to compute
    \begin{equation*}
        \int\int \alpha(u)-\gamma(x)\alpha(x)[dF_0(x,y)dK(u,v) + dH(x,y)dG_0(u,v)].
    \end{equation*}
    First compute
    \begin{equation*}
        \begin{split}
            \int\int \alpha(u)-\gamma(x)\alpha(x) dF_0(x,y)dK(u,v) &= \int\int (\alpha(u)-\gamma(x)\alpha(x)) f(x)k(u) dxdu\\
            &= \int \alpha(u) k(u) du - \int \gamma(x)\alpha(x)f(x)dx\\
            &= \int \alpha(u)k(u)du - \int \alpha(x)g(x)dx.
        \end{split}
    \end{equation*}
    The other term is
    \begin{equation*}
        \begin{split}
            \int\int \alpha(u)-\gamma(x)\alpha(x) dH(x,y)dG_0(u,v) &= \int\int (\alpha(u)-\gamma(x)\alpha(x)) h(x)g(u) dxdu\\
            &= \int \alpha(u) g(u) du - \int \gamma(x)\alpha(x)h(x)dx
        \end{split}
    \end{equation*}
    Putting these together, we see
    \begin{equation*}
        \begin{split}
            &\int\int \alpha(u)-\gamma(x)\alpha(x)[dF_0(x,y)dK(u,v) + dH(x,y)dG_0(u,v)]\\
            &= \int \alpha(u)k(u)du - \int \alpha(x)g(x)dx + \int \alpha(u) g(u) du - \int \gamma(x)\alpha(x)h(x)dx\\
            &= \int \alpha(u)k(u)du - \int \gamma(x)\alpha(x)h(x)dx.
        \end{split}
    \end{equation*}
    Thus,
    \begin{equation*}
        \frac{d}{d\epsilon}\U[\gamma_\epsilon(X)a(X,Y,U,V)] \big|_{\epsilon = 0} = \int\int \alpha(u)-\gamma(x)\alpha(x)[dF_0(x,y)dK(u,v) + dH(x,y)dG_0(u,v)]
    \end{equation*}
    so we have pathwise differentiability.
\end{proof}

\begin{proof}[Proof of Theorem \ref{thm test normal}]
    We just need to show that under the given conditions, the assumptions $\ref{as: consistency}$, \ref{as: sample ratio} and \ref{as: bias} are satisfied.
    To satisfy assumption \ref{as: consistency}, we need conditions so that $\|\hat{\psi}_{st} - \psi\| = o_p(1)$ for every cross-fit fold $(s,t)$. Let's expand this out as
\begin{align*}
        &\hat{\psi}_{st} - \psi\\
         = & (\hat{\gamma}(X) - \gamma(X))a(X,Y,U,V) + (\hat{\alpha}(U)-\alpha(U)) + (\alpha(X)\gamma(X) -  \hat{\alpha}(X) \hat{\gamma}(X))\\
        = & (\hat{\gamma}(X) - \gamma(X))a(X,Y,U,V) + (\hat{\alpha}(U)-\alpha(U)) + \gamma(X)(\hat{\alpha}(X)-\alpha(X))\\
        &~~~~~ + \hat{\alpha}(X)(\gamma(X)-\hat{\gamma}(X))\,.
\end{align*}

Note that as $|a(X,Y,U,V)| \leq 1$ and $|\hat{\alpha}(X)|\leq 1$ since $a$ takes on values $0,1$ and $\hat{\alpha}$ is a probability, requiring $\|\hat{\gamma}(X) - \gamma(X)\| = o_p(1)$ and $\|\hat{\alpha}(U) - \alpha(U)\| = o_p(1)$ takes care of the first, second and last terms. We just need to the term $\gamma(X)(\hat{\alpha}(X)-\alpha(X))$, which requires that $\|\gamma(X)\hat{\alpha}(X) - \gamma(X)\alpha(X)\| = o_p(1)$.


Finally, we need to analyze the remainder term from Assumption \ref{as: bias}. It is enough to show that for each fold $(s,t)$, $C_{st} = o_p\left(\frac{1}{\sqrt{m+n}}\right)$. We can expand
\begin{align*}
        C_{st} &= \U (\hat{\psi}_{st} - \psi) \\
        &= E_{F,G} \left[\hat{\gamma}(X)a(X,Y,U,V) + \hat{\alpha}(U) - \hat{\gamma}(X)\hat{\alpha}(X) - \gamma(X)a(X,Y,U,V)\right.\\
        &~~~~~~~~~~~~~\left.-\alpha(U)+\gamma(X)\alpha(X)\right]\\
        &= E_X \left[\hat{\gamma}(X)\alpha(X) + \gamma(X) \hat{\alpha}(X) - \hat{\gamma}(X)\hat{\alpha}(X) - \gamma(X)\alpha(X) \right]\\
        &= E_X\left[(\gamma(X) - \hat{\gamma}(X))(\hat{\alpha}(X) -\alpha(X))\right].
\end{align*}

In particular, we can bound $|C_{st}|$ by the product of errors $\|\hat{\gamma} - \gamma\|\cdot\|\hat{\alpha} - \alpha\|$. This gives the final condition.

This condition satisfies the double robust property that just the product of the nuisance rates has to reach the $\sqrt{m+n}$ rate. This is ideal as it allows the use of slower converging nonparametric estimators for the nuisance functions while still achieving normality for the parameter of interest at the $\sqrt{m+n}$ rate. 
\end{proof}

\end{document}